\documentclass[11pt,a4paper]{article}

\usepackage[left=1in,right=1.5in,top=1in,bottom=1in]{geometry}
\usepackage[T1]{fontenc}
\usepackage[utf8]{inputenc}

\usepackage{epsfig, float}
\usepackage{enumitem}
\usepackage{float}
\usepackage{lscape}
\usepackage{setspace}
\usepackage{graphicx}
\usepackage{graphics}
\usepackage{mathcomp}
\usepackage{dsfont}
\usepackage{lmodern}

\usepackage{amsmath,amssymb}
\usepackage{amsfonts}
\usepackage{url}
\usepackage{bbm}
\usepackage{mathrsfs} 
\usepackage[english]{babel}
\usepackage[hidelinks]{hyperref}
\usepackage{array}

\newcommand{\R}{\mathbb{R}}
\newcommand{\C}{\mathbb{C}}

\newcommand{\id}{\mathbbm{1}}

\newcommand{\diag}{{\rm diag} }
\newcommand{\dom}{{\rm dom}}

\newcommand{\vv}[1]{{\mathbf{#1}}}

\newtheorem{theorem}{Theorem}[section] 
\newtheorem{lemma}[theorem]{Lemma}
\newtheorem{proposition}[theorem]{Proposition}

\newtheorem{definition}[theorem]{Definition}

\newenvironment{proof}[1][Proof:]{\begin{trivlist}
\item[\hskip \labelsep {\bfseries #1}]}{\end{trivlist}}
\newenvironment{proofT}[1][Proof of Theorem 2.2:]{\begin{trivlist}
\item[\hskip \labelsep {\bfseries #1}]}{\end{trivlist}}
\newenvironment{example}[1][Example:]{\begin{trivlist}
\item[\hskip \labelsep {\bfseries #1}]}{\end{trivlist}}
 \newenvironment{remark}[1][Remark:]{\begin{trivlist}
 \item[\hskip \labelsep {\bfseries #1}]}{\end{trivlist}}

\newcommand{\qed}{\hfill\ensuremath{\square}}
\newcolumntype{C}[1]{>{\centering\let\newline\\\arraybackslash\hspace{0pt}}m{#1}}

\begin{document}
\title{Consistency of multi-time Dirac equations \\ 
    with general interaction potentials}

\author{
Dirk-Andr\'e Deckert\thanks{deckert@math.lmu.de} \ and  Lukas Nickel\thanks{nickel@math.lmu.de} \\[0.2cm]
	Mathematisches Institut, Ludwig-Maximilians-Universit\"at\\
	Theresienstr.\ 39, 80333 M\"unchen, Germany 
}

\date{\today}

\maketitle

\begin{abstract}
    In 1932, Dirac proposed a formulation in terms of multi-time wave functions
    as candidate for relativistic many-particle quantum mechanics. A 
    well-known consistency condition that is necessary for existence of solutions 		    strongly restricts the possible interaction types between the particles. It was
    conjectured by Petrat and Tumulka that interactions described by
    multiplication operators are generally excluded by this condition, and they
    gave a proof of this claim for potentials without spin-coupling.  Under
    smoothness assumptions of possible solutions we show that there are
    potentials which are admissible, give an explicit example, however, show
    that none of them fulfills the physically desirable Poincaré invariance.
    We conclude that in this sense Dirac's multi-time formalism does not allow
    to model interaction by multiplication operators, and briefly point out several 
    promising approaches to interacting models one can instead pursue.
\end{abstract}

\section{Introduction}

The absence of absolute simultaneity in the theory of relativity has
consequences for the formulation of relativistic quantum mechanics. Very
elementarily, this can already be observed when considering the Lorentz
transformation of a simultaneous configuration of $N$ particles, $(t,
\mathbf{x}_1),...,(t, \mathbf{x}_N)$, which yields a configuration $(t_1',
\mathbf{x}_1'),...,(t_N', \mathbf{x}_N')$ with $N$ different times. This fact
immediately poses the question of how a wave function or quantum state $\psi (t,
\mathbf{x}_1, ..., \mathbf{x}_N)$, which is usually described as dependent on
one time $t$ and Euclidean positions $\vv x_1,\dots,\vv x_N$, behaves under such
a transformation. Dirac addressed this issue already in 1932 and suggested to
generalize the concept of the familiar wave function $\psi (t, \mathbf{x}_1,
..., \mathbf{x}_N)$ to a \emph{multi-time} wave function $\psi (x_1, ..., x_N)$, where now $x_j= (t_j, \mathbf{x}_j)$ denote $N$ space-time points in Minkowski space; see \cite{DiracRQM}. This idea led to the fundamental works
\cite{DiracFockPodolsky,tomonaga_relativistically_1946} which provided the basis
for the relativistic formulation of quantum field theory.  In his approach,
Dirac defined the evolution of the multi-time state $\psi$ by requiring it to
fulfill $N$ Dirac equations, one for each each time variable $t_j$. Although this
concept seems natural, it is very restrictive in admission of solutions because it is a necessary condition for the existence of solutions, already discussed in
\cite{blochmultitime} and henceforth called \emph{consistency condition}, that the $N$ single-time evolutions commute. This condition becomes
subtle when the $N$ particles are allowed to interact. In this respect,
Dirac's approach calls for a mathematical study of the corresponding
solution theory, which was initiated recently in a series of works by Petrat and
Tumulka \cite{multitimeNoPotentials,multi-timeQFT,multitimePaircreation} and by
Lienert \cite{MatthiasPhD,1d_model,matthias2bd} and Lienert and Nickel
\cite{NTeilchenPaper}. As shown
in \cite{multitimeNoPotentials}, the consistency condition basically rules out
any interaction mediated by potentials without spin-coupling. In the following we extend the results of \cite{multitimeNoPotentials} and prove that the consistency condition is also violated for Poincar{\'e} invariant interaction potentials including
spin-coupling. 
\\This raises the question of how to introduce a sensible interaction in the multi-time formalism, which led Dirac et al.\ \cite{DiracFockPodolsky} to consider 
second-quantized fields that mediate the interaction; see also the
recently studied multi-time models of quantum field theory in
\cite{multi-timeQFT,multitimePaircreation}. In one dimension, another way of introducing a consistent interaction between the
$N$ particles was presented in \cite{1d_model,NTeilchenPaper}. There, rigorous models of interaction by boundary conditions have been constructed. There is some connection to the new concept of interior-boundary conditions by Teufel and Tumulka, which has so far been used to formulate certain non-relativistic QFT models without divergences \cite{IBC1,IBC2,IBC1d}. It is an open but very interesting question if the method of interior-boundary conditions can help to formulate mathematically well-defined models of particle creation and annihilation in the multi-time formalism.
\\ A further strategy that has been pursued is to generalize the concept of a potential to terms of the form $V(x_1,\dots,x_N,p_1,\dots,p_n)$ that are no multiplication operators, but also depend on the momenta, i.e.\ derivatives  \cite{twobodydirac1, twobodydirac2}. Lastly, we consider the idea of multi-time integral equations to be very promising. Instead of a system of differential equations such as \eqref{eq:intromultsys}, one can impose a single integral equation for $\psi(x_1,...x_N)$. This avoids the problem of the consistency condition and makes a more general class of models possible.
A prominent example known from QED is the Bethe-Salpeter equation \cite{bethesalpeter,bethesalpeterReview}, whose mathematical features are not well-understood and would deserve further study (see also \cite{master_nickel}).

\paragraph{Definition of the model.} The model for our investigation is given
by the system of evolution equations
\begin{equation} 
    \label{eq:intromultsys} i \frac{\partial}{\partial t_k} \psi (x_1, ..., x_N)
    = H_k(x_1,...,x_N)  \psi (x_1, ..., x_N), \quad k=1,...,N,
\end{equation}
where the \textit{partial Hamiltonians} $H_k$ are given by
\begin{equation} \label{eq:Hamiltonianform}
H_k = H^0_k + V_k,
\end{equation}
with $H^0_k$ being the free Dirac Hamiltonian of the $k$-th particle (see
\eqref{eq:freedirac}) below). The interaction shall be described by the operator 
$V_k$ which is given in terms of 
a (self-adjoint) spin-matrix valued multiplication
operator $V_k(x_1,...x_N)$ that depends on the space-time coordinates $x_1,
\dots, x_N$.  For this model, as was first recognized by Bloch
\cite{blochmultitime} and further investigated by Petrat and Tumulka
\cite{multitimeNoPotentials}, a necessary condition for existence of solutions to 
\eqref{eq:intromultsys} is the aforementioned \emph{consistency
condition}
\begin{equation} \label{eq:ccversion2}
    \left(\left[ H_j,  H_k \right]  - i \frac{\partial V_k}{\partial t_j} +  i
    \frac{\partial V_j}{\partial t_k}\right)\psi= 0, \qquad \forall k \neq j. 
\end{equation}
In \cite{multitimeNoPotentials}, Petrat und Tumulka conjectured that
interacting systems of the form \eqref{eq:intromultsys} with general
non-vanishing potentials that lead to interaction between the particles are
excluded as they would violate the consistency condition \eqref{eq:ccversion2}. They
gave a proof of this claim under
the assumption that the potentials $V_k$ depend on the spin-index of the $k$-th
particle only.  This rules out a number of
conceivable potentials, but not all of them: Potentials such as the one of the
Breit equation \cite{Breit1929, Breit1932}, which can be derived as an
approximation to the Bethe-Salpeter equation of QED (see \cite{greinerQED}),
contain a more complicated spin-coupling, which poses the question whether more
general potentials may indeed comply with condition \eqref{eq:ccversion2} and thereby to
well-posedness of \eqref{eq:intromultsys} in terms of an initial value problem.
\\

As main results of this paper, we present a concrete example of a spin-coupling
interaction potential which satisfies the consistency condition. However, we
will also show that the class of potentials admitted by the consistency
condition is rather small.  In particular, under certain smoothness conditions on
possible solutions $\psi$, we identify this class completely and show that it does not
contain Poincar{\'e} invariant potentials. Therefore, combining the mathematical consistency condition
with the physical requirement of Poincar{\'e} invariance, our results show that any
type of potential acting as a multiplication operator must be excluded
as possible candidates for modeling the interaction between the $N$ particles.
\\

After the following paragraph about the employed notation and conventions, we
present our results in Section~\ref{sec:results} and the proofs and more
detailed derivations in Sections~\ref{sec:cc} and \ref{sec:spinpotentials}. 

\paragraph{Notations and conventions.}
We consider $4$-dimensional Minkowski space-time with metric $g = \diag (1, -1,
-1, -1)$, with the usual notation that Greek indices run from $0$ to $3$ and
Latin indices $a,b,...$ only over the spatial components $1,2,3$. The Einstein
summation convention is employed for Greek indices only. Particle labels are
denoted also by Latin indices, $j,k,...$ and run from $1$ to the total particle
number $N$. Space-time points are denoted by $x = (t, \mathbf{x})$. Throughout,
the abbreviation $\partial_{k,\mu} := \frac{\partial}{\partial x_k^\mu}$ will be
used.

The gamma matrices are arbitrary $4 \times 4$-matrices that form a
representation of the Clifford algebra, i.e.\ fulfill the anti-commutation
relation 
\begin{equation}
\left\{ \gamma^\mu, \gamma^\nu \right\} = 2 g^{\mu \nu} \id, \quad \mu, \nu = 0,1,2,3.
\end{equation} 
Moreover, the matrix $\gamma^0$ is hermitian, $\gamma^k$ anti-hermitian, and a fifth gamma matrix is defined as
\begin{equation}
\gamma^5 := i \gamma^0 \gamma^1 \gamma^2 \gamma^3.
\end{equation}
The free Dirac Hamiltonian for the $k$-th particle is given by
 \begin{equation} \label{eq:freedirac}
H^0_k = -i  \sum_{a=1}^3 \gamma_k^0 ~ \gamma^a_k \partial_{k,a} + \gamma_k^0 m_k,
\end{equation}
where $m_k$ is the mass of the $k$-th particle and we use the following convention for the matrices: Since we are always working in the $N$-fold tensor product of $\C^4$, we write for some $4 \times 4$-matrix $M$:
\begin{equation}
M_k := \id \otimes \dots \otimes \id \underbrace{\otimes \ M \ \otimes}_{\text{k-th place}} \id \otimes \dots \otimes \id .
\end{equation}
It is well-known that the Dirac operator \eqref{eq:freedirac} is self-adjoint on
$\dom(H^0_k) = H^1(\R^3,\C^4)$; see \cite{thaller}.
Furthermore, it will be convenient to use the notation $\alpha_k^\mu :=
\gamma^0_k \gamma^\mu_k$ so that we may write the multi-time system 
\eqref{eq:intromultsys} as
\begin{equation} \label{eq:multitimediracwithalpha}
(i \alpha^\mu_k \partial_{k, \mu} - \gamma^0_k m_k) \psi (x_1, ..., x_N) =  V_k (x_1, ...x_N) \psi (x_1, ..., x_N), \quad k=1,...,N.
\end{equation}
Hence, the wave function $\psi (x_1, ..., x_N)$ takes values in $(\C^4)^{\otimes
N} \cong \C^K$, $K := 4^N$. 

\section{Results} \label{sec:results}

In order to present the results two remarks are in order.  First, we need to make
precise what is meant by
the notion \textit{interaction potential}. External potentials of the form $V_k(x_k)$
that do not generate entanglement must be excluded, and also potentials that
seemingly depend on different coordinates, but that actually only arise from
external potentials by a change of coordinates in the spinor space $\C^K$.
Therefore we define:
\begin{definition}
A collection of potentials $V_k$, $k=1,...,N,$ given as spin-matrix valued
multiplication operators $V_k(x_1,...,x_N)$ is called non-interacting iff there is
a unitary map $U(x_1,...,x_N): \C^K \to \C^K$ such that for all $k=1,\dots,N$,
$\tilde{\psi}:=U(x_1,...x_N)(\psi(x_1,...x_N))$ satisfies a system of the form \eqref{eq:intromultsys} where for each $k$, the potential $V_k(x_k)$ is independent of all other coordinates $x_1,...x_{k-1},x_{k+1},...x_N$. In the other case, we call the collection of potentials interacting.
\end{definition}
Petrat and Tumulka called potentials that are connected via a unitary map $U$
\textit{gauge-equivalent} \cite{multitimeNoPotentials}, which means that
interacting potentials in the sense of our definition are exactly those that are
not gauge-equivalent to external potentials.

Second, it has to be emphasized that 
the natural domain of a multi-time wave function is not the whole configuration
space-time $\R^{4N}$, but the subset 
\begin{equation}
    \label{eq:configurationspace}
\mathscr{S}^{(N)} := \left\lbrace \left. (t_1, \mathbf{x}_1,...,t_N,\mathbf{x}_N) \in \R^{4N} \right| \forall k \neq j: (t_j - t_k)^2 < |\mathbf{x}_j - \mathbf{x}_k|^2  \right\rbrace,
\end{equation} 
which contains the configurations where the $N$ particles are space-like
separated. A detailed explanation of this fact is found in \cite{1d_model}.
Here, we only state that there are at least two reasons to consider a multi-time
wave function only on $\mathscr{S}^{(N)}$:
\begin{itemize}
    \item \textit{Sufficiency}: In order to interpret Born's rule on any space-like hypersurface,
        it is sufficient for $\psi$ to have domain $\mathscr{S}^{(N)}$.
        A Lorentz transformation of a simultaneous configuration as presented
        above always yields a space-like configuration. Indeed, the mere concept
        of ``$N$-particle configuration'' implies the use of $\mathscr{S}^{(N)}$
        because the presence of $N$ particles is always understood with respect
        to a frame, e.g.\ a laboratory frame, which is represented by a
        space-like hypersurface.
    \item \textit{Necessity}: In quantum field theory the left-hand side of the
        consistency condition \eqref{eq:ccversion2} generically contains
        commutators of field operators, such as $[\phi(x_j),\phi(x_k)]$, which are
        given in terms of the Pauli-Jordan distribution
        \cite{tomonaga_relativistically_1946}. However, the latter has only
        support for $(x_j-x_k)^2\geq 0$, and hence, outside of $\mathscr{S}^{(N)}$.
        This is the reason why multi-time formulations of quantum field theory
        such as \cite{DiracFockPodolsky} as well as 
        \cite{multi-timeQFT, multitimePaircreation} are consistent on
        $\mathscr{S}^{(N)}$, but not on $\R^{4N}$. 
\end{itemize}
Therefore, all results will be proven mainly on 
$\mathscr{S}^{(N)}$ and only besides on $\R^{4N}$. Lastly, we have to make precise what is meant by Poincar{\'e} invariance of potentials. For $\Lambda$ in the proper Lorentz group and $a \in \R^4$, the Poincar{\'e} transformation maps $x \mapsto x' =  \Lambda x + a$ and the multi-time wave function transforms as
\begin{equation}
\psi'\left(x_1,...x_N \right) = S(\Lambda)^{\otimes N} \psi \left(\Lambda^{-1} (x_1-a),...,\Lambda^{-1} (x_N-a) \right),
\end{equation}
with the spin transformation matrix $S(\Lambda)$ that fulfills $S(\Lambda) \gamma S^{-1}(\Lambda) = \Lambda \gamma$ . We call a potential $V_k$ Poincar\'e invariant if it satisfies
\begin{equation}
V_k \left(x_1,...x_N \right) = S(\Lambda)^{\otimes N} V_k \left(\Lambda^{-1} (x_1-a),...\Lambda^{-1} (x_N-a) \right)  S^{-1} (\Lambda)^{\otimes N},
\end{equation}
which is the condition for \eqref{eq:intromultsys} to be Poincar{\'e} invariant.
Our main result can then be stated as follows:
\begin{theorem} \label{thm:main} \noindent
    Let $N=2$, $\Omega=\mathbb R^{4N}$ or $\Omega=\mathscr{S}^{(N)}$.
    If $V_k(x_1,\cdots,x_N)$ are interacting potentials in $C^1(\Omega, \C^{K \times K})$ and for all initial values $\varphi \in C^\infty_c (\mathbb R^{3N}\cap \Omega, \C^K)$, there is a solution $\psi \in C^2(\Omega,\C^{K})$ to the multi-time system of Dirac equations \eqref{eq:intromultsys}, then the potentials $V_k$ are not Poincar{\'e} invariant.
\end{theorem}
We only formulate the theorem for the case $N=2$, although we expect it to hold for general $N$ and we prove several intermediate results for any $N$.  For larger numbers of particles, however, some parts in the proofs which are based on a direct computation in terms of gamma matrices quickly become very complex and hardly traceable. In several partial results, we will also not restrict to $\Omega=\mathbb R^{4N}$ or $\Omega=\mathscr{S}^{(N)}$, but consider any open set $\Omega \subset \R^{4N}$. 
The strategy of proof is illustrated as follows:
\begin{enumerate}[label=(\textbf{\alph*})]
    \item \textbf{Existence $\Longrightarrow$ Consistency:} If a solution to
         \eqref{eq:intromultsys} exists, then the consistency
        condition \eqref{eq:ccversion2} has to hold.
    \item \textbf{Consistency $\Longrightarrow$ Restrictions on potentials:} If
        the consistency condition \eqref{eq:ccversion2} holds, then the admissible
        potentials are restricted and no Poincar{\'e} invariant ones are
        possible.
\end{enumerate}

\subsection*{Step (a): The consistency condition.} Let us first discuss why one
expects the consistency condition \eqref{eq:ccversion2} to be necessary for existence of solutions. The condition can heuristically be understood as path independence of the integration of the system of evolution equations  \eqref{eq:intromultsys}: E.g. prescribing initial values $\psi (0, \mathbf{x}_1,...0, \mathbf{x}_N) $ at $t_1 = ... = t_N = 0$, it makes no difference if one decides to evolve first in $t_j$-direction and then in $t_k$-direction or the other way around, one always has to arrive at the same well-defined $\psi (t_1, \mathbf{x}_1, ..., t_N, \mathbf{x}_N)$. Therefore, the actions of the respective equations in our system \eqref{eq:intromultsys} on the possible initial values have to commute. 

Petrat and Tumulka have proven  that the existence of a solution for every initial datum in the Hilbert space necessitates the consistency condition \eqref{eq:ccversion2} in two different cases \cite[Theorems 1 and 2]{multitimeNoPotentials}: 
\begin{itemize}
\item for time-independent, possibly unbounded partial Hamiltonians $H_k$,
\item for time-dependent, but smooth and bounded partial Hamiltonians $H_k$.
\end{itemize} 
Here, we generalize the results of Petrat and Tumulka to the relevant case of unbounded Hamiltonians that may include a time-dependence in the potentials. Our proposition is a rather direct consequence of the differentiability of solutions and makes the idea of Bloch \cite[p.~304]{blochmultitime} mathematically precise. 

\begin{proposition} \label{thm:ccdifferentiable}
Let $\Omega \subset \R^{4N}$ be open. Suppose the multi-time system \eqref{eq:intromultsys}, with $V_k$ being a function in $C^{1} (\Omega,\C^{K \times K})$, possesses a solution $\psi \in C^{2} (\Omega,\C^K)$. Then the consistency condition \eqref{eq:ccversion2} holds for all $(x_1,...,x_N) = X \in \Omega$. 
\end{proposition}

The proof is given in Section \ref{sec:proofccdifferentiable}, followed by some remarks about a more geometric way of understanding the consistency condition in Section \ref{sec:geometric}.




\subsection*{Step (b): Consistent potentials.} 
The consistency condition puts strong restrictions on the spin-coupling induced
by the potentials. The following example shows the inconsistency for one natural looking choice.
\begin{example}We consider a two-particle system  \eqref{eq:intromultsys} with $V_1 = \alpha_2^\mu A_\mu(x_1, x_2) $ and $V_2 = \alpha_1^\mu B_\mu(x_1, x_2)$ for some smooth, compactly supported functions $A_\mu, B_\mu$. This is suggested by the usual way of adding a 4-vector potential to the single-time Dirac equation, which is by adding $\alpha^\mu A_\mu$ to the Hamiltonian. One could think that interaction is achieved by choosing the gamma matrices of the other particle, as done here. But then the consistency condition is
\begin{align}
\left[ \alpha^\mu_2 A_\mu , -i \alpha^\nu_2 \partial_{2, \nu} + \gamma^0_2 m_2 \right] &= 0 \nonumber
\\ \Longleftrightarrow -2m_2 \gamma_2^\mu A_\mu + i \alpha_2^\nu \alpha_2^\mu \left( \partial_{2, \nu} A_\mu \right) + i A_\mu [ \alpha^\nu_2, \alpha^\mu_2] \partial_{2, \nu} &= 0.
\end{align}
There is no possibility that the respective terms will cancel each other, so any $A_\mu$ different from zero will make the equations inconsistent. In particular, the derivative term with $ \partial_{2, \nu}$ has to vanish separately, which will be a crucial ingredient in the proof of theorem \ref{thm:ccimplies}. A similar calculation  excludes potentials of the form $V_k \sim F_{\mu \nu}(x_1,x_2) \gamma_1^\mu \gamma_2^\nu$, too.
\end{example}

To have a chance of being consistent, the potentials may only depend on few matrices, which are the identity matrix and $\gamma^5$. To see this, we need to reformulate the consistency condition to a more useful version. That the bracket in \eqref{eq:ccversion2} applied to any solution $\psi$ ought to be zero implies that it must also be zero on every initial value $\varphi = \left.\psi \right|_{t_1=...=t_N=0}$. The initial values will be defined on a $3N$-dimensional set $U$, an intersection of $\Omega$ with the time-zero hypersurface. The assumption that there are solutions for all initial values in a certain class, e.g.\ the smooth compactly supported functions, allows us to draw general conclusions.

\begin{theorem} \label{thm:ccimplies} We assume:
\begin{enumerate}[label=(\textbf{\Alph*})]
\item  $U \subseteq \R^{3N}$ is open and simply connected. For a multi-time Dirac system \eqref{eq:intromultsys} with continuously differentiable $V_k$, we have for each $\varphi \in
C^\infty_c (U, \C^K)$,
\begin{equation} \label{eq:ccinitial}
\left( \left[ H_j,  H_k \right] - i  \frac{\partial V_k}{\partial t_j} +  i \frac{\partial V_j}{\partial t_k} \right) \varphi = 0, \qquad \forall k \neq j.
\end{equation}
\end{enumerate}
Then, for each $k \neq j$, the $k$-th spin component of the potential $V_j$ is spanned by $\id_k$ and $\gamma^5_k$.
\end{theorem}

The proof is given in Section \ref{sec:possiblespins}. One can directly see that the above example is not in the class of admissible potentials. 

Theorem \ref{thm:ccimplies} allows us to proceed by a basis decomposition. All possible matrix structures that might appear in $V_1$ and $V_2$ can be listed and the consistency condition can be explicitly evaluated, as will be done in Section \ref{sec:basisdecomp}. In Lemma \ref{thm:lemmacc}, we show that the consistency condition is equivalent to the system of equations \eqref{eq:cc1} to \eqref{eq:cc16}, and that only eight possibly interacting terms remain.
\\ It turns out that these possibilities for interacting terms in the potentials can not be excluded by general arguments. In fact, interacting potentials that fulfill the consistency condition exist, for example the ones in the following lemma.
\begin{lemma} \label{thm:consistentexample}
Let $C_\nu$ and $c_\nu$ be constants for $\nu = 0,1,2,3$ with at least one $C_\nu$ and $c_\nu$ different from zero, and define $x := x_2 - x_1$. Consider the multi-time Dirac system  \eqref{eq:intromultsys} for two particles with potentials 
\begin{align}
V_1 & = \gamma_1^\mu C_\mu \exp \left( 2i \gamma^5_1 c_\lambda x^\lambda \right) - m_1 \gamma^0_1 \nonumber
\\ V_2 & = \gamma_1^5 \alpha_2^\nu c_\nu . \label{eq:hoho}
\end{align}
\begin{enumerate}
\item This system is consistent, i.e.\ \eqref{eq:ccversion2} holds.
\item This system is interacting.
\end{enumerate}
\end{lemma}
This is proven in Section \ref{sec:consistentexample}. With this example at hand, it becomes clear that we cannot prove inconsistency of arbitrary interacting potentials. But obviously, the potential $V_1$ in \eqref{eq:hoho} is not Lorentz invariant. Since the use of multi-time equations aims at a relativistic formulation of quantum mechanics, it is natural to require Poincar{\'e} invariance, i.e.\ Lorentz invariance and translation invariance, of the potentials. We show that the latter excludes the former by finding that every translation invariant potential has to be of a certain shape.

\begin{lemma} \label{thm:potentialsform}
Suppose the assumptions $\mathbf{(A)}$ of theorem \ref{thm:ccimplies} hold. If, in addition, the potentials are both interacting and translation invariant, i.e. satisfy
\begin{equation}
V_k (x_1, x_2) = V_k (x_1 + a, x_2 + a) ~~ \forall a \in \R^4,
\end{equation}
then they are necessarily of the form 
\begin{equation} \label{eq:potentialsform}
V_k = M_1 e^{c_{k,\nu} x^\nu} + M_2 e^{-c_{k, \nu} x^\nu} + const.\
\end{equation}
for some $M_1, M_2 \in \C^{K \times K}$ and $c_k \in \C^4$, where $x = x_1 - x_2$. 
\end{lemma}

A slightly stronger version of this lemma will be formulated and proven in Section \ref{sec:consistentpotentials}. Our main theorem \ref{thm:main} can then be proven by a simple collection of facts:

\begin{proofT} \noindent
\begin{itemize}
\item First case: $\Omega = \R^{4N}$. Suppose a system \eqref{eq:intromultsys} with potentials $V_k \in C^1(\R^{4N},\C^{K \times K})$ that are interacting has a solution $\psi \in C^2(\R^{4N},\C^K)$ for all initial values $\varphi \in C^\infty_c (\R^{3N},\C^K)$. Consequently, by Proposition \ref{thm:ccdifferentiable}, the consistency condition \eqref{eq:ccinitial} has to be true for all $\varphi \in C^\infty_c (\R^{3N},\C^K)$. Then, by Lemma \ref{thm:potentialsform}, if the potentials are translation invariant, they are of the form \eqref{eq:potentialsform}, which is not Lorentz invariant. Therefore, the potentials cannot be Poincar{\'e} invariant.
\item Second case: $\Omega = \mathscr{S}^{(N)}$. The proof for the domain $\mathscr{S}^{(N)}$ goes through as above because the necessary lemmas were all proven for general domains that are open and simply connected, which is true for $\mathscr{S}^{(N)}$.
\end{itemize} \qed
\end{proofT}

Under the assumptions on higher regularity of solutions, we have thus generalized the results of Petrat and Tumulka \cite{multitimeNoPotentials} in the sense that our theorem \ref{thm:main} covers arbitrary multiplication operators with spin-coupling. The class of potentials that are consistent and translation invariant (equation \eqref{eq:potentialsform}) does not contain any physically interesting potentials, but only potentials that oscillate with the distance of the particles. That these are not Lorentz invariant further motivates to disregard them because multi-time equations are intended for a fully and manifest Lorentz invariant formulation of quantum mechanics.  The implications of this result for the formulation of interacting relativistic quantum mechanics were discussed above in the introduction.

\section{Proof of the consistency condition} \label{sec:cc}

\subsection{Proof of Proposition \ref{thm:ccdifferentiable}} \label{sec:proofccdifferentiable}

\begin{proof}[Proof of Proposition \ref{thm:ccdifferentiable}:] Suppose $\psi \in C^{2} (\Omega,\C^K)$ solves the equations \eqref{eq:intromultsys}. Let $j \neq k$. By the theorem of Schwarz, the time-derivatives on $\psi$ commute, which for $X \in \Omega$ gives:
\begin{align}
   \left( i \partial_{t_k} i \partial_{t_j} - i \partial_{t_j} i \partial_{t_k}\right) \psi = 0 
 \Rightarrow & ~ i \partial_{t_k} \left( H_j \psi  \right) - i \partial_{t_j} \left( H_k \psi  \right)  = 0 \label{eq:secondine}
\\ \Rightarrow & ~ H_j  i \partial_{t_k} \psi + \left(  i \partial_{t_k} V_j \right) \psi - \left(  i \partial_{t_j} V_k \right) \psi - H_k i \partial_{t_j} \psi  = 0 \label{eq:thirdline}
\\ \Rightarrow & ~ \left( H_j H_k +  \left(  i \partial_{t_k} V_j \right) - \left(  i \partial_{t_j} V_k \right) - H_k H_j \right) \psi  = 0. \label{eq:fourthline}
\end{align}
In \eqref{eq:secondine} and \eqref{eq:fourthline}, we used that $\psi$ solves the multi-time equations \eqref{eq:intromultsys}, and \eqref{eq:thirdline} follows by the product rule. As $X \in \Omega$ was arbitrary, equation \eqref{eq:ccversion2} holds on $\Omega$, as claimed. \qed
\end{proof}

\begin{remark}
\noindent
\begin{enumerate}
\item 
The assumption that the solution $\psi$ is at least twice differentiable in the time direction seems unproblematic because the spatial smoothness of initial data is usually inherited in the time direction due to the nature of physically relevant evolution equations. E.g.\ for the one-particle Dirac equation with smooth external electromagnetic potential $A_\mu$, it was proven in \cite{Deckert2014} that solutions that are smooth on one (space-like) Cauchy surface are indeed smooth on all of $\R^4$. 
\item This theorem even covers relativistic Coulomb potentials because for the domain $\Omega = \mathscr{S}^{(N)}$, a potential of the form
\begin{equation}
V \sim \frac{1}{(t_k - t_j)^2 - |\mathbf{x}_k - \mathbf{x}_j|^2}
\end{equation}
is singular only outside of $\mathscr{S}^{(N)}$, which ensures that $V \in C^{\infty} (\mathscr{S}^{(N)},\C^{K \times K})$.
\end{enumerate}
\end{remark}

\subsection{Geometric view of the consistency condition} \label{sec:geometric}
In this section, we discuss on a non-rigorous level how the results on the consistency condition can be reformulated with the help of differential geometry (compare Section 2.3 in \cite{multitimeNoPotentials}). For each multi-time argument $(t_1,...,t_N)$, the multi-time wave function will be an element of the Hilbert space $\mathscr{H} = L^2(\R^{3N},\C^K)$. We can define a vector bundle $E$ over the base manifold $\R^N$ with identical fibres $\mathscr{H}$ at every point. (This is therefore a trivial vector bundle $E = \R^N \times \mathscr{H}$). A multi-time wave function is then a section of $E$. 
\\ A natural notion of parallel transport on $E$ can be given by the single-time evolution operators $U_k(t_k)$ (which would be $e^{-iH_kt_k}$ for time-independent $H_k$). This means that we define a connection $\nabla$ on $E$ with components $\nabla_k = \partial_{t_k} + i H_k$, whereby the parallel transport in direction $t_k$ is given by $U_k$. Solutions of \eqref{eq:intromultsys} are then sections that are covariantly constant, i.e.\ satisfy $\nabla \psi=0$.
\\ The well-definedness of solutions requires that the parallel transport along a closed curve does not change the vector. So we need that for any loop $\gamma$, $U_\gamma = \id$. This is equivalent to saying that the vector bundle has a trivial holonomy group, $\mathrm{Hol}(\nabla) = \{ \id \}$. By the theorem of Ambrose and Singer \cite{ambrosesinger}, the holonomy group is in direct correspondence to the curvature form $F(\nabla)$; in particular: $\mathrm{Hol}(\nabla) = \{ \id \} \Leftrightarrow F(\nabla)=0.$ Therefore, the existence of a well-defined solution implies that $\nabla$ is a flat curvature for $E$. By the formula for calculating the curvature from the connection, this means
\begin{equation}
0 = F_{ij} = \frac{\partial H_i}{\partial t_j} -  \frac{\partial H_j}{\partial t_i} - i [ H_i, H_j],
\end{equation}
which is the consistency condition.

\section{Spin-coupling potentials}
\label{sec:spinpotentials}

\subsection{Proof of Theorem \ref{thm:ccimplies}}
\label{sec:possiblespins}

\begin{proof}[Proof of Theorem \ref{thm:ccimplies}:]
We start with a system \eqref{eq:multitimediracwithalpha} and evaluate the consistency condition \eqref{eq:ccinitial}. Let $k \neq j$, then:
\begin{align}
& \left[ i \alpha^\mu_k \partial_{k,\mu} - \gamma_k^0 m_k - V_k ,   i \alpha^\nu_j \partial_{j,\nu}- \gamma_j^0 m_j - V_j \right]
\\ = & \left[i \alpha^\mu_k \partial_{k,\mu} - \gamma^0_k m_k, -V_j \right]
+ \left[ -V_k, i \alpha^\nu_j \partial_{j,\nu} - \gamma^0_j m_j \right] + \left[ V_k, V_j \right] \label{eq:blablabla}
\\ = & \left[ V_k, V_j \right] + m_k \left[ \gamma_k^0, V_j \right] - m_j \left[ \gamma_j^0, V_k \right] -i \left[ \alpha^\mu_k \partial_{k,\mu}, V_j \right] +i
\left[ \alpha^\nu_j \partial_{j,\nu}, V_k \right]
\end{align}
In \eqref{eq:blablabla}, we used that the derivatives w.r.t.\ different coordinates commute by Schwarz. We consider the last term in more detail:
\begin{align}
i \left[ \alpha^\nu_j \partial_{j,\nu}, V_k \right] & = i \alpha^\nu_j \partial_{j,\nu} V_k - i V_k \alpha^\nu_j \partial_{j,\nu} \nonumber
\\ & = i \alpha^\nu_j  \left( \partial_{j,\nu}V_k \right) + i \alpha^\nu_j V_k  \partial_{j,\nu} - i V_k  \alpha^\nu_j \partial_{j,\nu} \nonumber
\\ & =   i \alpha^\nu_j  \left( \partial_{j,\nu}V_k \right) + i \left[ \alpha^\nu_j, V_k \right] \partial_{j,\nu} \nonumber
\\ & =   i \alpha^\nu_j  \left( \partial_{j,\nu}V_k \right)+ i \sum_{a=1}^3 \left[\alpha^a_j, V_k \right] \partial_{j,a},
\end{align}
where in the last line, the summand with $\nu = 0$ was dropped because $\alpha^0 = \id$ commutes with everything. Doing the same for the second last term yields that the consistency condition is equivalent to
\begin{align} \label{eq:explicitconscon}
0 = & \left[ V_k, V_j \right] + m_k \left[ \gamma_k^0, V_j \right] - m_j \left[ \gamma_j^0, V_k \right]  \nonumber
\\ - & i \alpha^\mu_k  \left( \partial_{k,\mu}V_j \right) +  i \alpha^\nu_j  \left( \partial_{j,\nu}V_k \right)  \nonumber
\\ - & i \sum_{a=1}^3 \left[ \alpha^a_k, V_j \right] \partial_{k,a} +  i \sum_{a=1}^3 \left[ \alpha^a_j, V_k \right] \partial_{j,a}.
\end{align}
The derivatives in \eqref{eq:explicitconscon} are in some sense linearly independent, which is made clear in the following auxiliary claim.
\begin{lemma}
Let $U \subseteq \R^{3N}$ be open. Let $f: U \rightarrow \C^K$ be a function and suppose there are complex $K \times K$-matrices $\Lambda_{k,j}(\mathbf{x}_1,...,\mathbf{x}_N)$ such that
\begin{equation} \label{eq:lemmalin}
\left( f(\mathbf{x}_1,...,\mathbf{x}_N) + \sum_{k=1}^{N} \sum_{j=1}^{3} \Lambda_{k,j} \frac{\partial}{\partial x_k^j} \right) \varphi(\mathbf{x}_1,..,.\mathbf{x}_N) = 0, \quad \forall (\mathbf{x}_1,...,\mathbf{x}_N) \in U,
\end{equation}  
holds for all $\varphi \in C^{\infty}_c (U,\C^K)$. Then, for all $j$ and $k$, $\Lambda_{k,j}(\mathbf{x}_1,...,\mathbf{x}_N) =0$, and $f$ must be the zero function.
\end{lemma}
\begin{proof}[Proof of the Lemma:]
We choose some fixed $k$ and $j$ and show that $\Lambda_{k,j} =0$ first. Pick some point $(\mathbf{x}_1,...,\mathbf{x}_N)=\mathbf{X} \in U$. There exists $\varphi \in C^{\infty}_c (U,\C^K)$ with the property that $\varphi(\mathbf{X}) = 0$ and $\partial_{l,m} \varphi(\mathbf{X}) = \delta_{lk} \delta_{mj}$. Thus, evaluating \eqref{eq:lemmalin} at the point $\mathbf{X}$, we have
\begin{equation}
0 = f(\mathbf{X}) \varphi(\mathbf{X}) + \sum_{l=1}^{N} \sum_{m=1}^{3} \Lambda_{l,m}(\mathbf{X})  \delta_{lk} \delta_{mj} = \Lambda_{k,j}(\mathbf{X})
\end{equation}
Because all factors $\Lambda_{k,j}$ are equal to zero, eq.\ \eqref{eq:lemmalin} directly implies that $f$ is the zero function. \qed
\end{proof}
Applying this lemma to the consistency condition \eqref{eq:explicitconscon}, we obtain that the prefactors of the derivative terms have to vanish separately, which means
\begin{equation} \label{eq:commutatoralpha}
\left[ \alpha^a_j, V_k \right]=0, \quad \forall k \neq j , ~ \forall a \in \{ 1,2,3\}.
\end{equation}
This will give us the desired constraint on the matrix structures that may appear in each $V_k$. We note that the following matrices form a basis of the complex $4 \times 4$ matrices (for a proof see e.g.\ \cite[p.~53ff.]{Bogoliubov}) : 
\begin{equation} \label{eq:basis}
\alpha^\mu, \quad \gamma^5 \alpha^\mu, \quad \gamma^\mu,\quad \gamma^5 \gamma^\mu ,\qquad \mu= 0,1,2,3.
\end{equation}
Although the matrix $V_k$ is a tensor product of $N$ $4 \times 4$-matrices, we can disregard all factors of the tensor product apart from the $j$-th to check when the condition \eqref{eq:commutatoralpha} can be satisfied. We can express $V_k$ in the above basis and just compute all commutators of $\alpha^a$ with basis elements. The following list, where we omit the index $j$, results:
\begin{align}
\left[ \alpha^a, \alpha^0 \right] & = 0 \nonumber
\\ \left[ \alpha^a, \alpha^b \right] & = 2 \gamma^a \gamma^b = -2i \varepsilon_{abc} \gamma^5 \alpha^c \nonumber
\\ \left[ \alpha^a, \gamma^5 \alpha^0 \right] & = 0 \nonumber
\\ \left[ \alpha^a,  \gamma^5 \alpha^b \right] & = (2 - 2\delta^{ab}) \gamma^5 \gamma^b \gamma^a = 2 i \varepsilon_{abc}  \alpha^c \nonumber
\\ \left[ \alpha^a, \gamma^0 \right] & = -2 \gamma^a \nonumber
\\ \left[ \alpha^a, \gamma^b \right] & = -2 \delta^{ab} \gamma^0 \nonumber
\\ \left[ \alpha^a, \gamma^5 \gamma^0 \right] & = -2 \gamma^5 \gamma^a \nonumber
\\ \left[ \alpha^a, \gamma^5 \gamma^b \right] & = -2 \delta^{ab} \gamma^5 \gamma^0.
\end{align}
If $V_k$ contains combinations of $\alpha^0_j = \id_j$ and $\gamma^5_j$, the commutators in \eqref{eq:explicitconscon} vanish. But the commutators with all other elements of the basis give non-zero and linearly independent matrices, which implies that other matrices cannot be present in $V_k$ in order for condition \eqref{eq:commutatoralpha} to be fulfilled.  \qed
\end{proof}

\subsection{Basis decomposition} \label{sec:basisdecomp}

By theorem \ref{thm:ccimplies}, the consistency condition implies that $V_k$ only depends on the spin of the $j$-th particle via the identity matrix or $\gamma^5_j$. Therefore, we can expand the potentials as
\begin{align}
V_1 & = \id_2 V_{11} + \gamma^5_2 V_{15}, \nonumber
\\ V_2 &= \id_1 V_{21} + \gamma^5_1 V_{25}.
\end{align}
In the terms $V_{i1}$ and $V_{i5}$, all matrices depending on the $i$-th spin index may appear in principle, so we have
 \begin{align}
V_{11}  & = \alpha_1^\mu W_{1,\mu}  + \gamma^5_1 \alpha_1^\mu Y_{1,\mu} + \gamma_1^\mu A_\mu + \gamma^5_1 \gamma_1^\mu B_\mu  \nonumber \\
V_{15}  & = \alpha_1^\mu X_{1,\mu} + \gamma^5_1 \alpha_1^\mu Z_{1,\mu} + \gamma_1^\mu C_\mu + \gamma^5_1 \gamma_1^\mu D_\mu  \nonumber \\
V_{21}  & = \alpha_2^\nu W_{2,\nu} + \gamma^5_2 \alpha_2^\nu X_{2,\nu}+ \gamma_2^\nu E_\nu + \gamma^5_2 \gamma_2^\nu F_\nu \nonumber  \\
V_{25} & = \alpha_2^\nu Y_{2,\nu} + \gamma^5_2 \alpha_2^\nu Z_{2,\nu} + \gamma_2^\nu G_\nu + \gamma^5_2 \gamma_2^\nu H_\nu , \label{eq:expansionofpotentials}
\end{align}
where $A_0, B_k, C_0, D_k, E_0, F_k, G_0, H_k, W_{i,\mu},X_{i,\mu},Y_{i,\mu},Z_{i,\mu}$ are arbitrary real scalar functions and $A_k, B_0, C_k, D_0, E_k, F_0, G_k, H_0$ are arbitrary functions with purely imaginary values, such that the potentials are self-adjoint. It will soon become understandable why this nomenclature makes sense, especially what $W_1,X_1,Y_1,Z_1$ have to do with $W_2,X_2,Y_2,Z_2$.

\begin{lemma} \label{thm:lemmacc}
Consider a multi-time system \eqref{eq:intromultsys} for two particles for which the assumption \textbf{(A)} of Theorem \ref{thm:ccimplies} holds. Then the potentials can be expanded as
\begin{align} \label{eq:joerg}
V_1 &=   \gamma_1^\mu A_\mu + \gamma^5_1 \gamma_1^\mu B_\mu + \gamma_2^5 \left( \gamma_1^\mu C_\mu + \gamma^5_1 \gamma_1^\mu D_\mu \right) + V_{1,ext}
\\ V_2 &=\gamma_2^\nu E_\nu + \gamma^5_2 \gamma_2^\nu F_\nu + \gamma^5_1 \left( \gamma_2^\nu G_\nu + \gamma^5_2 \gamma_2^\nu H_\nu \right) + V_{2,ext} \label{eq:heike}
\end{align}
where $V_{i,ext}$ is not interacting and the functions $A_\mu$ to $H_\mu$, $\mu=0,1,2,3$, are scalars. 
\\ Furthermore, the consistency condition is equivalent to the following system of equations:
\begin{subequations}
\begin{align} 
 \partial_{1, \mu} W_{2, \nu}&= \partial_{2, \nu} W_{1, \mu}  \label{eq:cc1} \\
\partial_{1, \mu} X_{2, \nu} &= \partial_{2, \nu} X_{1, \mu}   \label{eq:cc2}\\  
\partial_{1, \mu} Y_{2, \nu} &= \partial_{2, \nu} Y_{1, \mu}  \label{eq:cc3}
 \\  \partial_{1, \mu} Z_{2, \nu} &= \partial_{2, \nu} Z_{1, \mu}  \label{eq:cc4}
\\  \label{eq:cc5}
B_\mu Y_{2, \nu} +  D_\mu Z_{2, \nu}&=  \tfrac{i}{2} \partial_{2,\nu} A_\mu   
\\  \label{eq:cc6} (m_1 \delta_{0 \mu } +A_\mu) Y_{2, \nu} + C_\mu Z_{2, \nu}&=  \tfrac{i}{2} \partial_{2,\nu}  B_\mu 
\\  \label{eq:cc7} -B_\mu Z_{2, \nu} -  D_\mu Y_{2, \nu} &=  \tfrac{i}{2} \partial_{2,\nu} C_\mu 
 \\ -(m_1  \delta_{0 \mu} +A_\mu) Z_{2, \nu} -  C_\mu Y_{2, \nu}&=  \tfrac{i}{2} \partial_{2,\nu} D_\mu \label{eq:cc8}  \\  \label{eq:cc9}
 F_\nu X_{1, \mu} +  H_\nu Z_{1, \mu}&= \tfrac{i}{2} \partial_{1, \mu} E_\nu 
 \\ (m_2 \delta_{0\nu} +  E_\nu) X_{1, \mu}  +  G_\nu Z_{1, \mu} &= \tfrac{i}{2} \partial_{1, \mu} F_\nu  \label{eq:cc10}
 \\  -F_\nu Z_{1,\mu}- H_\nu X_{1,\mu}&= \tfrac{i}{2} \partial_{1, \mu} G_\nu   \label{eq:cc11}
 \\ -(m_2 \delta_{0\nu} +  E_\nu) Z_{1, \mu}   -  G_\nu X_{1, \mu} &=  \tfrac{i}{2} \partial_{1, \mu} H_\nu  \label{eq:cc12}
\\ B_\mu G_\nu &=  C_\mu F_\nu   \label{eq:cc13}
\\   B_\mu H_\nu &= C_\mu (m_2\delta_{0\nu} + E_\nu )  \label{eq:cc14}
 \\ (m_1  \delta_{0 \mu } +A_\mu) G_\nu &= D_\mu F_\nu  \label{eq:cc15}
 \\ (m_1\delta_{0 \mu} + A_\mu ) H_\nu &=  D_\mu (m_2 \delta_{0\nu}+ E_\nu)  \label{eq:cc16}
\end{align}
\end{subequations}
\end{lemma}

\begin{proof}[Proof of Lemma \ref{thm:lemmacc}:]
Having used Theorem \ref{thm:ccimplies} already and expanded the potentials as in \eqref{eq:expansionofpotentials}, we now evaluate the missing part of the consistency condition:
\begin{equation}
0 \stackrel{!}{=} \left[ V_k, V_j \right] + m_k \left[ \gamma_k^0, V_j \right] - m_j \left[ \gamma_j^0, V_k \right]  -  i \alpha_k^\mu \partial_{k, \mu} V_j  +  i \alpha^{\mu}_j \partial_{j,\mu} V_k
\end{equation}
We have
\begin{align}
m_1 \left[ \gamma_1^0, V_2 \right] & = 2 m_1    \gamma^0_1 \gamma^5_1 V_{25}, \nonumber
\\ m_2 \left[ \gamma_2^0, V_1 \right] & = 2 m_2 \gamma^0_2 \gamma^5_2 V_{15},
\end{align}
and 
\begin{align}
 \left[ V_1, V_2 \right]  = & \left[ V_{11}, \gamma_1^5 \right] V_{25} +  \left[ \gamma_2^5, V_{21} \right] V_{15} +  \left[ V_{15} \gamma_2^5, \gamma^5_1 V_{25}  \right]  \nonumber
 \\  =  & \left( -2\gamma_1^5 \gamma_1^\mu A_\mu - 2 \gamma^\mu_1 B_\mu  \right) V_{25}  \nonumber
 \\ & + \left( 2\gamma_2^5 \gamma_2^\nu E_\nu + 2 \gamma^\nu_2 F_\nu  \right) V_{15}  \nonumber
 \\ & +2 \alpha^\mu_1 \gamma_1^5 X_{1,\mu} \left( \gamma^5_2 \gamma_2^\nu G_\nu + \gamma_2^\nu H_\nu \right) + 2 \alpha^\mu_1 Z_{1, \mu}\left( \gamma^5_2 \gamma_2^\nu G_\nu + \gamma_2^\nu H_\nu \right)  \nonumber \\ & + 2 \gamma_1^\mu \gamma_1^5 C_\mu \left( \gamma^5_2 \alpha^\nu_2 Y_{\nu,2} + \alpha^\nu_2 Z_{2, \nu} \right) - 2\gamma_1^\mu D_\mu \left( \gamma^5_2 \alpha^\nu_2 Y_{\nu,2} + \alpha^\nu_2 Z_{2, \nu} \right).
\end{align}
The derivative terms are
\begin{equation} \label{eq:dervs}
- i \alpha_1^\mu \partial_{1, \mu} V_{21}   - i \alpha_1^\mu \gamma_1^5 \partial_{1, \mu} V_{25}+  i \alpha^{\nu}_2 \partial_{2,\nu} V_{11} +  i \alpha^{\nu}_2  \gamma^5_2 \partial_{2,\nu} V_{15}.
\end{equation}

As the 16 matrices in \eqref{eq:basis} are linearly independent, their tensor products give us $16^2 = 256$ linearly independent matrices that appear in the consistency condition. Their respective prefactors have to vanish separately. This gives the following table, in which every of the $16$ cells stands for $16$ terms (for $\mu,\nu = 0,1,2,3$) that have to vanish.

\vspace*{0.75cm}
\noindent \begin{tabular}{c||C{2.89cm}|C{2.89cm}|C{2.4cm}|C{2.9cm}}
$\otimes$ & $\alpha_2^\nu$ & $\gamma^5_2 \alpha_2^\nu$ & $\gamma_2^\nu$ &  $ \gamma_2^5 \gamma_2^\nu$ \\ \hline \hline
$\alpha_1^\mu$ & $-i \partial_{1, \mu} W_{2, \nu}+ i \partial_{2, \nu} W_{1, \mu}$ & $-i \partial_{1, \mu} X_{2, \nu} + i \partial_{2, \nu} X_{1, \mu}$ & $2F_\nu X_{1, \mu} + 2 H_\nu Z_{1, \mu}- i \partial_{1, \mu} E_\nu  $ & $(2m_2 \delta_{0\nu} + 2 E_\nu) X_{1, \mu}  + 2 G_\nu Z_{1, \mu} - i \partial_{1, \mu} F_\nu$ 
 \\ \hline
$\gamma^5_1 \alpha_1^\mu$ & $-i \partial_{1, \mu} Y_{2, \nu} + i \partial_{2, \nu} Y_{1, \mu}$ & $-i \partial_{1, \mu} Z_{2, \nu} + i \partial_{2, \nu} Z_{1, \mu}$  & $2F_\nu Z_{1,\mu} + 2H_\nu X_{1,\mu}+i \partial_{1, \mu} G_\nu$ & $(2m_2 \delta_{0\nu} + 2 E_\nu) Z_{1, \mu}   + 2 G_\nu X_{1, \mu} + i \partial_{1, \mu} H_\nu$
\\ \hline
$\gamma_1^\mu  $ & $-2B_\mu Y_{2, \nu} - 2 D_\mu Z_{2, \nu}+ i \partial_{2,\nu} A_\mu$ & $-2B_\mu Z_{2, \nu} - 2 D_\mu Y_{2, \nu}- i \partial_{2,\nu} C_\mu$ & $-2B_\mu G_\nu + 2 C_\mu F_\nu$ & $-2 B_\mu H_\nu + 2 E_\nu C_\mu+ 2m_2 C_\mu \delta_{0\nu}$  
  \\ \hline
$\gamma^5_1 \gamma_1^\mu$ &  $-(2m_1 \delta_{0 \mu } +2A_\mu) Y_{2, \nu}- 2 C_\mu Z_{2, \nu}+ i \partial_{2,\nu}  B_\mu$ & $-(2m_1  \delta_{0 \mu} +2A_\mu) Z_{2, \nu} - 2 C_\mu Y_{2, \nu}- i \partial_{2,\nu} D_\mu$ & $-(2m_1  \delta_{\mu 0} +2A_\mu) G_\nu+2 D_\mu F_\nu $ & $-(2 A_\mu + 2m_1\delta_{0 \mu} ) H_\nu +(2m_2 \delta_{0\nu}+ 2 E_\nu) D_\mu$  
\\   \hline
\end{tabular}
\vspace*{0.5cm}

Setting every entry of this table equal to zero gives the required system of equations \eqref{eq:cc1}--\eqref{eq:cc16}.
\\ It remains to show that the potentials can be expanded as in \eqref{eq:joerg}, \eqref{eq:heike}. Let us add up equations \eqref{eq:cc1} to \eqref{eq:cc4} with the respective matrices, factorizing $\alpha_1^\mu \alpha_2^\nu$, which leads to
\begin{align} \label{eq:eichweg}
& - \partial_{1, \mu} W_{2, \nu}+  \partial_{2, \nu} W_{1, \mu} + \gamma^5_2 \left( - \partial_{1, \mu} X_{2, \nu}+  \partial_{2, \nu} X_{1, \mu} \right) \nonumber \\ & + \gamma^5_1 \left( - \partial_{1, \mu} Y_{2, \nu}+  \partial_{2, \nu} Y_{1, \mu} \right) + \gamma^5_1 \gamma^5_2 \left( - \partial_{1, \mu} Z_{2, \nu}+  \partial_{2, \nu} Z_{1, \mu} \right) = 0.
\end{align}
The names we gave to the terms in the potential are suited to make the symmetry of this equation visible. Defining 
\begin{equation}
f_{j,\mu} := W_{j, \mu} + \gamma^5_2 X_{j, \mu} + \gamma^5_1 Y_{j, \mu}  + \gamma^5_1 \gamma^5_2 Z_{j, \mu},
\end{equation}
equation \eqref{eq:eichweg} becomes
\begin{equation} \label{eq:Fschief}
\partial_{1, \mu} f_{2, \nu} = \partial_{2, \nu} f_{1, \mu}.
\end{equation}
Then, we adapt the argument of Petrat and Tumulka \cite[p.~34]{multitimeNoPotentials}: Define 
\begin{equation}
g_{j,\mu\nu} = \partial_{j,\mu}f_{j,\nu} - \partial_{j,\nu}f_{j,\mu}.
\end{equation}
For $i \neq j$, we have
\begin{equation} \label{eq:klaus}
\partial_{i,\lambda}g_{j,\mu\nu} =  \partial_{j,\mu}\partial_{i,\lambda}f_{j,\nu} -   \partial_{j,\nu}\partial_{i,\lambda} f_{j,\mu} = \partial_{j,\mu} \partial_{j,\nu} f_{i,\lambda} - \partial_{j,\nu} \partial_{j,\mu}  f_{i,\lambda}= 0.
\end{equation}
This implies that $g_{j,\mu\nu}$ is a function of $x_j$ only. Define for arbitrary fixed $\tilde{x}_1, \tilde{x}_2$ the function $\tilde{f}_{j,\mu} (x_j):=f_{j,\mu}(x_j, \tilde{x}_i)$ and $h_{j,\mu}(x_1,x_2):=f_{j,\mu}(x_1,x_2)-\tilde{f}_{j,\mu}(x_j)$. Since \eqref{eq:klaus} implies
\begin{equation}
g_{j,\mu\nu} = \partial_{j,\mu}f_{j,\nu} - \partial_{j,\nu}f_{j,\mu} = \partial_{j,\mu}\tilde{f}_{j,\nu} - \partial_{j,\nu}\tilde{f}_{j,\mu},
\end{equation}
we have
\begin{equation}
\partial_{j,\mu}h_{j,\nu} - \partial_{j,\nu}h_{j,\mu} = 0, \quad j=1,2.
\end{equation}
Moreover, eq. \eqref{eq:Fschief} gives us 
\begin{equation}
\partial_{1,\mu}h_{2,\nu} - \partial_{2,\nu}h_{1,\mu} = 0.
\end{equation}
These two equations together form the integrability condition, from which it follows that a self-adjoint matrix-valued function $M(x_1,x_2)$ exists such that $h_{j,\mu} = \partial_{j,\mu} M(x_1,x_2)$, i.e.\
\begin{equation}
f_{j, \mu} (x_1, x_2) =  \partial_{j,\mu} M(x_1,x_2)  + \tilde{f}_{j,\mu} (x_j).
\end{equation}
Therefore, the unitary map $e^{i M(x_1,x_2)}$ maps the potential $f_j$ to the purely external potential $\tilde{f}_j$, which shows that $f_j$ is not interacting according to our definition.
\\ The generalization to the case where the consistency condition only holds on $\mathscr{S}^{(N)}$ works exactly like in \cite[p.~35]{multitimeNoPotentials}. \qed

\end{proof}

\subsection{A consistent example} \label{sec:consistentexample}

As a side remark before we prove Lemma \ref{thm:consistentexample}, note that the connection of the consistent potential with the above basis decomposition is more easily visible if the potential is rewritten as $V_1 = -i \gamma_1^\mu C_\mu \sin (2c_\nu x^\nu)  + \gamma_1^5 \gamma_1^\mu C_\mu \cos (2c_\nu x^\nu)- m_1 \gamma^0_1$.

\begin{proof}[Proof of Lemma \ref{thm:consistentexample}:] \noindent \begin{enumerate}
\item
We have to evaluate the consistency condition
\begin{align*}
\left[ i \alpha_1^\mu \partial_{1,\mu} - m_1 \gamma^0_1 - \gamma_1^\mu C_\mu \exp \left( 2i \gamma^5_1 c_\lambda x^\lambda \right) + m_1 \gamma^0_1, i \alpha_2^\nu \partial_{2,\nu} - m_2 \gamma^0_2 - \gamma_1^5 \alpha_2^\nu c_\nu \right] & = 0
\\ \Longleftrightarrow  - \left[  \gamma_1^\mu C_\mu \exp \left( 2i \gamma^5_1 c_\lambda x^\lambda \right), i \alpha_2^\nu \partial_{2,\nu} \right] + \left[  \gamma_1^\mu C_\mu \exp \left( 2i \gamma^5_1 c_\lambda x^\lambda \right), \gamma_1^5 \alpha_2^\nu c_\nu \right] & = 0
\\ \Longleftrightarrow \gamma_1^\mu C_\mu \alpha_2^\nu \left( i\partial_{2, \nu}  \exp \left( 2i \gamma^5_1 c_\lambda x^\lambda \right) + 2\gamma^5_1 c_\nu \exp \left( 2i \gamma^5_1 c_\lambda x^\lambda \right) \right)  &= 0 ,
\end{align*}
which is indeed true. Note that in the case at hand the consistency condition is satisfied identically, not only applied to certain functions. 
\item 
Now we assume (for a contradiction) that there is a gauge transformation $U(x_1,x_2): \C^K \to \C^K$ that yields non-interacting potentials. Such a map can be written as $U(x_1, x_2) = e^{iM(x_1,x_2)}$ with a self-adjoint $K \times K$-matrix $M$. We define the transformed quantities
\begin{equation}
\tilde{\psi} := U \psi, ~~ \tilde{\gamma^\mu} := U \gamma^\mu U^{\dagger}.
\end{equation}
If $\psi$ is a solution of the system \eqref{eq:intromultsys}, it follows that $\tilde{\psi}$ satisfies
\begin{equation}
(i \tilde{\alpha}^\mu_k \partial_{k, \mu} - \tilde{\gamma^0_k} m_k) \tilde{\psi} = \tilde{V}_k \tilde{\psi} - \tilde{\alpha}^\mu_k (\partial_{k,\mu} \tilde{M}) \tilde{\psi},
\end{equation}
where $\tilde{V}$ and $\tilde{M}$ stand for the same expressions as $V$ and $M$, but with all appearing matrices replaced by the ones with a tilde\footnote{Since the gamma matrices are always only defined up to a similarity transformation, the tildes do not really matter and can basically be omitted. Note that a gauge transformation just refers to a (local) change of coordinates in the spinor space.}. Therefore, the condition that the transformed potential only depends on $x_k$ amounts to the requirement that
\begin{equation}
V_k(x_1,x_2) - \alpha^\mu_k \partial_{k,\mu}M(x_1,x_2)
\end{equation}
is in fact only a matrix-valued function of $x_k$, so its derivative with respect to another coordinate has to vanish. Using that $V_2$ is constant, this implies the following two equations:
\begin{align} \label{eq:mama}
\partial_{1, \lambda} \alpha^\mu_2 \partial_{2,\mu} M(x_1,x_2) & = 0
\\ \partial_{2, \delta} \alpha^\nu_1 \partial_{1,\nu} M(x_1,x_2) & = c_\delta 2i \gamma^5_1 \gamma^\mu_1 C_\mu \exp \left( 2i \gamma^5_1 c_\nu x^\nu \right)  \label{eq:papa}
\end{align}
Now consider the contraction 
\begin{align}
& \alpha^\lambda_1 \alpha^\delta_2 \partial_{1, \lambda}\partial_{2, \delta} M(x_1,x_2) \nonumber
\\ ~=~ & \alpha^\lambda_1 \left( \alpha^\delta_2 \partial_{1, \lambda}\partial_{2, \delta} M(x_1,x_2) \right) = 0 \nonumber
\\ ~=~ & \alpha^\delta_2 \left( \alpha^\lambda_1 \partial_{1, \lambda}\partial_{2, \delta} M(x_1,x_2) \right) = \alpha^\delta_2 c_\delta 2i \gamma^5_1 \gamma^\mu_1 C_\mu \exp \left( 2i \gamma^5_1 c_\nu x^\nu \right) 
\end{align}
where we have used, after different regrouping of the summands, equation \eqref{eq:mama} in the second line and \eqref{eq:papa} in the third line.  This is a contradiction because the $C_\mu, c_\mu$ are not all zero. Hence, a matrix $M$ with the required properties does not exist. We have therefore proven that the potential is not gauge-equivalent to a non-interacting one, so it is interacting. \qed
\end{enumerate}
\end{proof}

\subsection{Classification of consistent potentials}
\label{sec:consistentpotentials}


Instead of proving lemma \ref{thm:potentialsform} directly, we give a slightly stronger reformulation that implies it, but uses the basis decomposition discussed in Section \ref{sec:basisdecomp}. 

\begin{lemma} \label{thm:lol}
Suppose the consistency condition is fulfilled (in the sense of \textbf{(A)} in theorem \ref{thm:ccimplies}) for a two-particle Dirac system \eqref{eq:intromultsys} for which the gauge transformation which makes $W_i,X_i,Y_i,Z_i$ purely external has already taken place. If the potentials are translation invariant, i.e. satisfy
\begin{equation}
V_i (x_1, x_2) = V_i (x_1 + a, x_2 + a) ~~ \forall a \in \R^4,
\end{equation}
then all terms $A_\mu,...,H_\mu$ in the potentials are necessarily of the form 
\begin{equation} \label{eq:formofpots}
C_1 \cdot e^{c_{i,\nu} x^\nu} + C_2 \cdot e^{-c_{i, \nu} x^\nu} 
\end{equation}
for some $C_1, C_2 \in \C$ and $c_i \in \C^4$, where $x = x_1 - x_2$. In the case of $A_0$ and $E_0$, a constant term $- m_1$ resp.\ $-m_2$ is added. 
\end{lemma}

\begin{proof}[Proof of Lemma \ref{thm:lol}]
After the gauge transformation, $W_i, X_i, Y_i$ and $Z_i$ are functions of $x_i$ only. If we assume that the potentials are translation invariant, it follows that these functions have to be constants. Therefore, we can derive second order differential equations for the functions $A$ to $H$. We show the steps for $B_\mu$ and $D_\mu$, the other cases are analogous. Since $V_k \in C^1(\Omega,\C^{K \times K})$, every scalar function $A_\mu, B_\mu,...,H_\mu,W_{i,\mu},...,Z_{i,\mu}$ in the potentials has to be continuously differentiable. Equations \eqref{eq:cc5} to \eqref{eq:cc12} imply that the terms $A$ to $H$ are in fact two times continuously differentiable, because the first derivatives are expressible as a sum of continuously differentiable functions.
\\ Therefore, we may differentiate equation \eqref{eq:cc6} once more. Inserting \eqref{eq:cc5} and \eqref{eq:cc7}, we obtain
\begin{equation}
\frac{1}{4} \partial_{2, \nu} \partial_{2, \lambda} B_\mu = (Z_{2,\lambda} Z_{2,\nu} - Y_{2,\lambda} Y_{2,\nu}) B_\mu +(Y_{2,\nu}Z_{2,\lambda} - Y_{2,\lambda}Z_{2,\nu}) D_\mu.
\end{equation}
Similarly for $D_\mu$:
\begin{equation}
\frac{1}{4} \partial_{2, \nu} \partial_{2, \lambda} D_\mu = (Z_{2,\lambda} Z_{2,\nu} - Y_{2,\lambda} Y_{2,\nu}) D_\mu +(Y_{2,\nu}Z_{2,\lambda} - Y_{2,\lambda}Z_{2,\nu}) B_\mu
\end{equation}
Although the derivatives $\partial_{2, \nu}$ and $\partial_{2, \lambda}$ need to commute, the right hand side of these equations is apparently not invariant under exchange of $\nu$ and $\lambda$. This implies that 
\begin{equation}
B_\mu = D_\mu =0 ~ \vee ~ Y_{2,\nu}Z_{2,\lambda} - Y_{2,\lambda}Z_{2,\nu} = 0.
\end{equation}
In the first case, we are already done (the potentials are of the desired form, with the constants being equal to zero). So we go on with the second case, where the differential equation becomes
\begin{equation}
 \partial_{2, \nu} \partial_{2, \lambda} B_\mu = 4(Z_{2,\lambda} Z_{2,\nu} - Y_{2,\lambda} Y_{2,\nu}) B_\mu,
\end{equation}
and the same for $D_\mu$. Using $Y_{2,\nu}Z_{2,\lambda} = Y_{2,\lambda}Z_{2,\nu}$, it can be rewritten as 
\begin{equation}
 \partial_{2, \nu} \partial_{2, \lambda} B_\mu = 2 \sqrt{Z_{2,\nu}^2 - Y_{2,\nu}^2} \cdot 2 \sqrt{Z_{2,\lambda}^2 - Y_{2,\lambda}^2} \cdot B_\mu.
\end{equation}
The square root is also defined for negative radicand as $\sqrt{x} := i \sqrt{|x|}$. This has the general solution
\begin{equation}
B_\mu = C_\mu^+ \exp \left(2 \sqrt{Z_{2,\alpha}^2 - Y_{2,\alpha}^2} x_2^\alpha \right) + C_\mu^- \exp \left(-2 \sqrt{Z_{2,\alpha}^2 - Y_{2,\alpha}^2} x_2^\alpha \right) ,
\end{equation}
with free constants $C_\mu^{\pm}$ that may depend on $x_1$. 
Since the potential must be translation independent, the constants must be such that $B_\mu$ has the form \eqref{eq:formofpots}.
\\ We thus have the required form for $B$ and $D$, and the other terms work analogously. In the case of $A$ and $E$, one should derive the differential equations for the functions $(m_1  \delta_{0 \mu} +A_\mu)$ and $(m_2  \delta_{0 \nu } +E_\nu)$ instead. Then, the consistency condition poses several additional constraints, eqs. \eqref{eq:cc13}--\eqref{eq:cc16} amongst others, that were not considered so far. But we will not elucidate on that because we only want to show that the form \eqref{eq:formofpots} is \textit{necessary}. \qed 

\end{proof}

\section*{Acknowledgements}
The authors are grateful for fruitful discussions with Matthias Lienert and S\"oren Petrat and also want to thank Matthias Lienert for helpful comments on the manuscript. This work was partly funded by the Elite Network of Bavaria through the Junior Research Group ‘Interaction between Light and Matter’.

\bibliography{Konsistenz}

\end{document}